\documentclass{article}

\usepackage{arxiv}

\usepackage[utf8]{inputenc} 
\usepackage[T1]{fontenc}    
\usepackage{hyperref}       
\usepackage{url}            
\usepackage{booktabs}       
\usepackage{amsfonts}       
\usepackage{nicefrac}       
\usepackage{microtype}      
\usepackage{enumitem}
\usepackage{lineno,hyperref}
\usepackage{amsmath}
\modulolinenumbers[5]
\usepackage{amsthm}
\usepackage{dirtytalk}
\usepackage{siunitx}
\usepackage{graphicx}

\usepackage[numbers]{natbib}
\newtheorem{theorem}{Theorem}[section]

\newtheorem{lemma}[theorem]{Lemma}

\title{Random Norming Aids Analysis of Non-linear Regression Models with Sequential Informative Dose Selection}

\author{
Zhantao Lin \\
  Department of Statistics\\
  George Mason University\\
  Fairfax, VA 22030 \\
  \texttt{} \\
   \And
 Nancy Flournoy \\
  Department of Statistics\\
  University of Missouri\\
  Columbia, MO 65203\\
  \texttt{} \\
     \And
 William F. Rosenberger \\
  Department of Statistics\\
  George Mason University\\
  Fairfax, VA 22030 \\
  \texttt{} \\
}

\begin{document}
\maketitle

\begin{abstract}
A two-stage adaptive optimal design is an attractive option for increasing the efficiency of clinical trials. In these designs, based on interim data, the locally optimal dose is chosen for further exploration, which induces dependencies between data from the two stages. When the maximum likelihood estimator (MLE) is used under nonlinear  regression  models  with  independent  normal  errors in a pilot study where the first stage sample size is fixed, and the second stage sample size is large, the Fisher information fails to normalize the estimator adequately asymptotically, because of dependencies. In this situation, we present three alternative random information measures and show that they provide better normalization of the MLE asymptotically. The performance of random information measures is investigated in simulation studies, and the results suggest that the observed information performs best when the sample size is small. 
\end{abstract}



\section{Introduction}
Two-stage designs are used for many purposes, including enrichment, sample size re-estimation and to modify randomization probabilities to improve the efficiency and/or efficacy of estimators. All these procedures use accumulated data to change the operation of the experimental design, which induces dependencies between the first and second stage data. Our interest lies in the effects of such dependencies on inference at the end of a pilot study where the first stage sample size is fixed, and the second stage sample size is large.

In two-stage enrichment designs, patients more likely to benefit from the treatment are identified based on data from the first stage, and second stage trials are conducted in the identified subpopulation [e.g., \citet{simon2004evaluating}, \citet{Anastasia:roy}, \citet{rosenblum:van}, \citet{trippa:rosner}, \citet{zang:guo}]. 
Two-stage sample size re-estimation methods are conducted by revising the final sample size with parameter estimation from the first stage [e.g., \citet{Stei:Atwo:1945}, \citet{proschan}, \citet{shih}, \citet{schwartz:denne}, \citet{zhong:carlin}, \citet{tarimainterim}, \citet{Broberg2017}]. In two-stage adaptive optimal designs, information from the first stage is used to estimate optimal treatment assignment probabilities for the second stage [e.g., \citet{haines2003bayesian}, \citet{lane:flou:two-:2012}, \citet{englert2013optimal}, \citet{lane:flou:info:2014}, \citet{shan2016optimal}].

\citet{lane:flou:two-:2012} studied asymptotic distributional properties of the maximum likelihood estimator for nonlinear regression models with independent normal errors.  In their study, they used the Fisher information to norm the score function when taking limits, obtaining a limiting distribution for the maximum likelihood estimator that is a random scale mixture of normal random variable. Use of this result requires knowledge of the distribution of the limiting scaling random variable. \citeauthor{lane:flou:two-:2012} found this distribution in the special case of an exponential mean function. But the method used is not generalizable, and so their result is informative, but not generally useful in practice.

In their review paper on likelihood theory for stochastic processes, \citet{Barn:Sore:arev:1994} describe conditions under which maximum likelihood estimators normed with the Fisher information converge to randomly scaled mixture of normal distributions, as was the case in \citet{lane:flou:two-:2012}. Limiting random mixtures of normal random variables also arise in \citet{Ivan:Rose:Durh:Flou:birt:2000}, \citet{ivanova:flournoy}, and \citet{may:flournoy}.  But \citeauthor{Barn:Sore:arev:1994} describe a solution to this problem. Namely, they describe how using a random norming in lieu of the Fisher information can lead to a standard normal distribution instead. 

This paper examines the use of random normings in a practical situation.
In particular, we evaluate these alternative random norms in the same context as in \citet{lane:flou:two-:2012} and \citet{lane:flou:info:2014}, and show how to apply them to obtain the more useful standard normal distribution. Then we 
compare the rates of convergence and efficiencies of the different norming alternatives.

Accordingly, this paper is organized as follows. In Section 2, we present the model to be studied in this paper. In Section 3, we describe stable and mixing convergences, which are needed, and a generalized version of the Cram{\'e}r-Slutzky theorem. In Section 4, we present the main asymptotic results for maximum likelihood estimators with random normings. We conduct simulation studies to compare the  efficiencies obtained with these normings for exponential and logistic models in Section 5.

\section{The Model}
Let $\{y_{ij}\}$ be observations from a two-stage adaptive design, where $n_{i}$ is the number of observations and $x_{i}\in[a,b]$ is the single-dose used for the $i$th stage, $i=1,2$. To avoid degenerate cases, we assume $n_i\ge 1, i=1,2$, and set $n=n_1+n_2$. We consider a general regression model with independent normal errors:
\begin{equation}\label{eq:model}
y_{ij}=\eta(x_{i},\theta)+\epsilon_{ij},\quad\epsilon_{ij}\sim N(0,\sigma^2),
\end{equation}
where $\eta(x_{i},\theta)=E(y_{ij}\vert x_{i})$ is some (possibly) nonlinear mean function, twice differentiable by $\theta$; $x_{1}\in[a,b]$ is given; and for simplicity, $\theta$ is a 1-dimensional parameter. In addition, adaptation is restricted to the choice of  $x_{2}$,  and $x_{2}$ depends on stage~1 data only through sufficient statistics from stage~1.
More specifically, $x_{2}=x_{2}(\bar{y}_1$) is a random function, where $\bar{y}_{j}=(y_{11}+\cdots+y_{1,n_j})/n_j$. Define $u_i=\sqrt{n_i}[\bar{y}_{i}-\eta(x_i,\theta)]/\sigma$. 
Then $u_1\sim N(0,1)$. 
But $u_2=u_2(\bar{y}_1)$. As for $u_2$, $E[u_2]=E_{\bar{y}_1}E[u_2\vert \bar{y}_1]=0$ and
 $Var[u_2]=E_{\bar{y}_1}E[u_2^2\vert \bar{y}_1]=1$.  But $u_2$
is only $N(0,1)$ conditionally on $\bar{y}_{1}$.  

Let $\hat{\theta}_{n_i}$ denote maximum likelihood estimators of $\theta$ based on stage~$i$  data, $i=1,2$, and let $\hat{\theta}_{n}$ denote the maximum likelihood estimator of $\theta$ based on all $n$ trials. 
 Since maximum likelihood estimators (MLEs) are functions of sufficient statistics, $\hat{\theta}_{n_1}$ is a function of the first stage mean response $\bar{y}_1$, and both  $\hat{\theta}_{n_2}$ and  $\hat{\theta}_{n}$ are functions of $(\bar{y}_1,\bar{y}_2)$.

Then the likelihood function is
\begin{align*}
\mathcal{L}_n(\theta|y_{11},\ldots,y_{1,n_1},y_{21},\ldots,y_{2,n_2})&=f_{n}(y_{11},\ldots,y_{1,n_1},y_{21},\ldots,y_{2,n_2}|\theta)\\
&=f_{n_1}(y_{11},\ldots,y_{1,n_1}|\theta)f_{n_2}(y_{21},\ldots,y_{2,n_2}|\theta,y_{11},\ldots,y_{1,n_1})\\
&\propto\exp\left\{-\tfrac{1}{2\sigma^2}\sum\limits^{n_1}_{i=1}[y_{1,i}-\eta(x_1,\theta)]^2
-\tfrac{1}{2\sigma^2}\sum\limits^{n_2}_{j=1}[y_{2,j}-\eta(x_2,\theta)]^2\right\}\\
&\propto\exp\left\{-\tfrac{n_1}{2\sigma^2}[\bar{y}_{1}-\eta(x_1,\theta)]^2
-\tfrac{n_2}{2\sigma^2}[(\bar{y}_{2}-\eta(x_2,\theta)]^2\right\}.
\end{align*}
Letting $\dot{\eta}(x_{i},\theta)=\tfrac{d}{d\theta}\eta(x_{i},\theta)$, and $\ddot{\eta}(x_i,\theta)=\frac{d^2}{d\theta^2}\eta(x_i,\theta)$, the score function can be written as
$$S_n(\theta)=\tfrac{d}{d\theta}\log \mathcal{L}_n(\theta)
=\tfrac{n_1}{\sigma^2}[\bar{y}_{1}-\eta(x_1,\theta)]\dot{\eta}(x_1,\theta)
+\tfrac{n_2}{\sigma^2}[\bar{y}_{2}-\eta(x_2,\theta)]\dot{\eta}(x_2,\theta).$$



\section{Stable and Mixing Convergence}
\subsection{Motivation and Definitions}
Let $\{X_{n}\}_{n\ge 1}$ and $X$ be real random variables defined on some probability space $(\Omega , \mathcal{F} , P)$, and let $\mathcal{G}\subset\mathcal{F}$ be a sub$-\sigma -$field.
Given a sequence of random variables $\{Y_{n}\}_{n\ge 1}$, suppose one wants to obtain the limiting distribution of the product of $Y_{n}X_{n}$.
 If $Y_{n}$ converges in probability to a constant $c$, and $X_{n}$ converges in distribution to $X$, then $Y_{n}X_{n} $ $\stackrel{d}{\rightarrow} cX \textrm{ as } n\rightarrow\infty$ by the Cram{\'e}r-Slutzky theorem [\cite{slutsky1925uber}].  
 However,
 \citet{lane:flou:two-:2012} showed for model \eqref{eq:model} that if $n_1/n$ is small (and provided common regularity conditions with $\dot{\eta}\equiv\dot{\eta}(x_{i},\theta)\neq 0,\ |\dot{\eta}|<\infty$), then 
 \begin{equation}\label{eq:lane}
 \sqrt[]{n}(\hat{\theta}_{n}-\theta)
 \approx U V_{n_2},
  \end{equation}
where $V_{n_2}=[\bar{y}_{2}-\eta(x_2,\theta)]\ /(\sigma \ / \ \sqrt[]{n_{2}})\sim Z$ for every $n_2$, where $Z\sim N(0,1)$; and
$U=\sigma[\dot{\eta}(x_{2},\theta)]^{-1}$ is independent of $V_{n_2}$ and $Z$. Since Equation \eqref{eq:lane} holds for all $n_2$, it holds in the limit as $n_2\rightarrow\infty$ with $n_1$ fixed. That is, $\sqrt[]{n}(\hat{\theta}_{n}-\theta)\stackrel{d}{\rightarrow}UZ$  as   $n_2\rightarrow\infty$  with $n_1$  fixed and $U$ independent of $Z$.  But $U$ is a random function of $\bar{y}_{1}$ that does not converge to a constant when $n_1$ is held fixed. So one cannot divide both sides of Equation \eqref{eq:lane} by $U$ and apply the classical Cram{\'e}r-Slutzky theorem  to obtain a  $N(0,1)$ limit.

To obtain a standard normal limit instead of the normal mixture in Equation \eqref{eq:lane} requires a generalized version of the Cram{\'e}r-Slutzky theorem, which is given in Lemma~\ref{con:function} below.  The generalized Cram{\'e}r-Slutzky theorem requires the concepts of stable and mixing convergence, which were introduced by \citet{Renyi1963bis}. So before proceeding, we recall these concepts. A thorough description of stable and mixing convergence can be found in \citet{Haus:Lusc:Stab:2015}.

Let $P_E(\cdot)=P(\ \cdot \ \cap E)/P(E)$ denote the conditional probability given the event $E$. We say that  $\{X_{n}\}$ converges $\mathcal{G}-$stably to $X$ as $n\rightarrow\infty$ 
if
\begin{align}\label{def:stable}
(X_{n})_{n\ge 1} \stackrel{d}{\rightarrow} X \textrm{ under } P_E \textrm{ for every event }  E\in \mathcal{G} \textrm{ with } P(E)>0.
\end{align}
Stable convergence is stronger than convergence in distribution, but not as strong as convergence in probability.
If $X$ is independent of $\mathcal{G}$, then the limit is said to be \emph{mixing}.

\subsection{Stable Convergence  Under Model \eqref{eq:model}.}
Under model \eqref{eq:model}, $\mathcal{F}_{n_1}=\sigma(\bar{y}_{1})$ is a sub~$\sigma-$field of $\mathcal{F}=\mathcal{F}_{n_1+n_2}=\sigma(\bar{y}_1,\bar{y}_2)$. 
In Lemma~\ref{stable}, we show that the convergence given in \eqref{eq:lane} is, in fact, stable convergence.  In the context of Equation~(\ref{def:stable}), take $\mathcal{G}=\mathcal{F}_{n_1}$.

\begin{lemma}\label{stable}
If $\dot{\eta}\neq 0$ and $ \vert\dot{\eta}\vert<\infty$ under model~\ref{eq:model}, 
$\sqrt[]{n}(\hat{\theta}_{n}-\theta)\rightarrow UZ\ \mathcal{F}_{n_1}-$stably with $U$ independent of $Z$ as $n_2\rightarrow\infty$ while $n_1$ is fixed.

\end{lemma}
\begin{proof}

Because of Equation \eqref{eq:lane} we only have to show that 
\[
U V_{n_2}\rightarrow UZ\quad \mathcal{F}_{n_1}-stably\mbox{ as } n_2\rightarrow\infty \mbox{ with } n_1 \mbox{ fixed}. \tag{*}
\]
Let the event $E\in\mathcal{F}_{n_1}$ with $P(E)>0$. By $\mathcal{F}_{n_1}$-measurability of $U$ and independence of $V_{n_2}$ and $\mathcal{F}_{n_1}$ and of $Z$ and $\mathcal{F}_{n_1}$ in combination with $V_{n_2}\sim Z$, we have $V_{n_2}U\sim ZU$ under $P_{E}$ for all $n_2$. In particular,
$$U V_{n_2}\stackrel{d}{\rightarrow} UZ \mbox{ under } P_{E} \mbox{ as } n_2\rightarrow\infty.$$
This proves $(*)$ in view of the definition of stable convergence in Equation \eqref{def:stable}

\end{proof}
\section{Standard normal limits with random norming}
\subsection{Random Norms and Their Limits under Model~\eqref{eq:model}}
\citet{Barn:Sore:arev:1994} describe random measures of information that can be used as norms for estimator and test statistics, and sometimes yield a more useful limit (e.g., standard normal) for MLEs. Following \citet{Barn:Sore:arev:1994}, we call them the observed, incremental observed and incremental expected information measures. In the two-stage setting, it not only makes sense to define increments in the log-likelihood  between individual subjects, but also  between stages because sufficient statistics are stage-wise data summaries. We examine both. 

First we formally define these, together with the expected (Fisher) information, and then we evaluate them under model \ref{eq:model}:
\begin{enumerate}
\item 
The {\bf observed information} is the negative derivative of the score function: 
\begin{align*}
j_n(\theta)&=-\dot{S}_n(\theta).
\end{align*}
\citet{Barn:Sore:arev:1994} and others have considered the observed information to be a standard with which the other information measures are  compared. 
\end{enumerate}

\begin{enumerate}
\setcounter{enumi}{1}
\item {\bf The Fisher information} is the variance of the score function. Assuming the integral and derivatives exist and are interchangeable, it is given by
\begin{align*}
i_n(\theta)&\equiv Var[S_n(\theta)]=E[S_n(\theta)]^2=E[-\dot{S}_n(\theta)].
\end{align*}

\citet{efron1978assessing}  studied the trade-off between the observed and expected (Fisher) information. They argue for using the observed information for data analysis after a study is completed, and they express a preference for using the expected information to design an experiment. \citet{Barn:Sore:arev:1994} state that \say{the difference (between the observed and expected information) is due, essentially, to the high content of ancillary information carried by the observed information.}
\citet{pierce_diss} and  \citet{firth1993bias} showed the observed information is larger than the Fisher information by an amount $O_p(n^{-1/2})$.

\qquad To define the incremental information in general, suppose a study is conducted in $K$ stages with $n_k$ subjects in each stage, $k=1,\ldots,K$. Then the log-likelihood can be written in increments as $
S_{n}(\theta)=\tfrac{d}{d\theta}\log \mathcal{L}_n(\theta)
=\sum_{i=1}^n\mathcal{D}_{i}(\theta)=\sum_{k=1}^K D_k(\theta),
$
where $\mathcal{D}_{i}(\theta)=S_{i}(\theta)-S_{i-1}(\theta)$ is the $i$th subject-wise increment and $D_k(\theta)=S_{n_1,\ldots, n_k}(\theta)-S_{n_1,\ldots, n_{k-1}}(\theta)$ is  the $k$th stage-wise increment with
$S_{0}(\theta)\equiv 0$.

\end{enumerate}

\begin{enumerate}
\setcounter{enumi}{2}
\item The {\bf incremental expected information} was introduced as the \emph{conditional variance} by \citet{levy1954theorie} in an early version of the Martingale central limit theorem. 
Let $\mathcal{F}_{i}$ denote the history of the experiment up through the trial for subject $i$, $i=1,\dots,n$; and let $\mathcal{F}_{0}$ be the trivial field. Then 
 $\mathcal{F}_{i}$ is a filtration of $\mathcal{F}$, i.e.: $\mathcal{F}_{0} \subset\mathcal{F}_{1} \subset\mathcal{F}_{2}\cdots \subset\mathcal{F}_{n}$. Using subject-wise and stage-wise increments in $S_{n}(\theta)$, we obtain the subject-wise and stage-wise incremental norms:
\begin{align*}
I_n^{\mathcal{D}}(\theta)
&=\sum\limits_{i=1}^{n}E_{\theta}[\mathcal{D}_i(\theta)^2|\mathcal{F}_{i-1}];\hspace{2cm}
I_n^{D}(\theta)=\sum\limits_{k=1}^{K}E_{\theta}[D_k(\theta)^2|\mathcal{F}_{n_{k-1}}].
\end{align*}
The incremental expected information  is also called the \emph{quadratic characteristic} of the score martingale.
\end{enumerate}

\begin{enumerate}
\setcounter{enumi}{3}
\item The {\bf incremental observed information} is given by
\begin{align*}
J^{\mathcal{D}}_n(\theta)
&=\sum\limits_{i=1}^{n}\mathcal{D}_i(\theta)^2;\hspace{2cm}
J^{D}_n(\theta)=\sum\limits_{k=1}^{K}D_k(\theta)^2.
\end{align*}
In the terminology of martingale theory, it is  called the \emph{quadratic variation}  of the score martingale [e.g., \citet{Barn:Sore:arev:1994}] and 
 \emph{squared variation}
[e.g., \citet{Hall:Heyd:Mart:1980}].
\citeauthor{Barn:Sore:arev:1994} show that use of the incremental observed information may  improve the robustness of estimators.
\end{enumerate}
It is common for the random information measures to converge to the Fisher information. However, there can be substantial differences with small sample sizes. Note that only observed and expected information are defined solely in terms of the likelihood function and its distribution law.  The incremental observed and expected information require knowledge of how the log-likelihood function increases from one subject or one stage to the next.

We now evaluate the random information norms that we will use to obtain standard normal limits for  $\hat\theta$. Under model \eqref{eq:model}, with $K=2$, the observed information is 
\begin{align}\label{obs}
j_n(\theta)
&=\frac{n_1}{\sigma^2}[\dot{\eta}(x_1,\theta)]^2
-\frac{n_1}{\sigma^2}\left[\bar{y}_1-\eta(x_1,\theta)\right]\ddot{\eta}(x_1,\theta)
+\frac{n_2}{\sigma^2}[\dot{\eta}(x_2,\theta)]^2
-\frac{n_2}{\sigma^2}\left[\bar{y}_2-\eta(x_2,\theta)\right]\ddot{\eta}(x_2,\theta).
\end{align}
The subject-wise and stage-wise incremental observed information are, respectively, 
\begin{align}\label{iobsub}
J_n^{\mathcal{D}}(\theta)&=\sum\limits^{n}_{i=1}[\mathcal{D}_i(\theta)]^2
=\sum\limits^{n_1}_{i=1}\frac{[y_i-\eta(x_1,\theta)]^2}{\sigma^4}[\dot{\eta}(x_1,\theta)]^2+\sum\limits^n_{i=n_1+1}\frac{[y_i-\eta(x_2,\theta)]^2}{\sigma^4}[\dot{\eta}(x_2,\theta)]^2
\end{align}
and
\begin{align}\label{iobs}
J^D_n(\theta)&=[D_1(\theta)]^2+[D_2(\theta)]^2=n_1^{2}\frac{[\bar{y}_1-\eta(x_1,\theta)]^2}{\sigma^4}[\dot{\eta}(x_1,\theta)]^2+n_2^{2}\frac{[\bar{y}_2-\eta(x_2,\theta)]^2}{\sigma^4}[\dot{\eta}(x_2,\theta)]^2.
\end{align}
The subject-wise and stage-wise incremental expected information are the same: 
\begin{align}\label{iexsub}
I^{\mathcal{D}}_n(\theta)&=\sum\limits^{n_1}_{i=1}E_{\theta}[\mathcal{D}_{i}(\theta)^2|\mathcal{F}_{i-1}]+\sum\limits^n_{i=n_1+1}E_{\theta}[\mathcal{D}_{i}(\theta)^2|\mathcal{F}_{i-1}]\nonumber\\
&=\sum\limits^{n_1}_{i=1}E_{\theta}\left\{\frac{(y_i-\eta(x_1,\theta))^2}{\sigma^4}[\dot{\eta}(x_1,\theta)]^2|\mathcal{F}_{i-1}\right\}+\sum\limits^n_{i=n_1+1}E_{\theta}\left\{\frac{(y_i-\eta(x_2,\theta))^2}{\sigma^4}[\dot{\eta}(x_2,\theta)]^2|\mathcal{F}_{i-1}\right\}\nonumber\\
&=\frac{{n}_{1}}{\sigma^2}[\dot{\eta}(x_1,\theta)]^2+\frac{{n}_{2}}{\sigma^2}[\dot{\eta}(x_2,\theta)]^2;
\end{align}
\begin{align}\label{iex}
I^{D}_n(\theta)&=E_{\theta}\{D_1(\theta)^2|\mathcal{F}_{0}\}+E_{\theta}\{D_2(\theta)^2|\mathcal{F}_{n_1}\}\nonumber\\
&=n_1^{2}E_{\theta}\left\{\frac{[\bar{y}_1-\eta(x_1,\theta)]^2}{\sigma^4}[\dot{\eta}(x_1,\theta)]^2|\mathcal{F}_{0}\right\}+n_2^{2}E_{\theta}\left\{\frac{[\bar{y}_2-\eta(x_2,\theta)]^2}{\sigma^4}[\dot{\eta}(x_2,\theta)]^2|\mathcal{F}_{n_1}\right\}\nonumber\\
&=\frac{{n}_{1}}{\sigma^2}[\dot{\eta}(x_1,\theta)]^2+\frac{{n}_{2}}{\sigma^2}[\dot{\eta}(x_2,\theta)]^2.
\end{align}
 Lemma \ref{cong:rv} provides convergence results for the random normings that are then used to obtain the desired standard normal limit for $\hat{\theta}_{n}$.
\begin{lemma}\label{cong:rv}
Under model \eqref{eq:model}, if $|\eta|<\infty$ and $|\dot{\eta}|<\infty$, as $n_2\rightarrow\infty$ with $n_1$ is fixed, 
\begin{enumerate}[label=(\arabic*)]
\item 
\begin{center} $n^{-1}j_n(\hat{\theta}_{n})\xrightarrow{p}U^{-2},$
\end{center}
\item 
\begin{center} $n^{-1}J^{\mathcal{D}}_n(\hat{\theta}_{n})\xrightarrow{p}U^{-2},$
\end{center}
\item 
\begin{center} $n^{-1}J^{D}_n(\hat{\theta}_{n})\xrightarrow{d}U^{-2}W,$
\end{center}
\item \begin{center} $n^{-1}I^{\mathcal{D}}_n(\hat{\theta}_{n})=n^{-1}I^{D}_n(\hat{\theta}_{n})\xrightarrow{p}U^{-2},$
\end{center}
\end{enumerate}
where $U^{-2}\equiv [\dot{\eta}(x_2,\theta)]^2\sigma^{-2}$ and 
$W\sim \chi^2(1)$.
\end{lemma}
\begin{proof}
\begin{enumerate}[label=(\arabic*)]
\item 

The first two terms of equation \eqref{obs} go to $0$ when divided by $n$. Note in the forth term that
$$\bar{y}_2-\eta(x_2,\theta)\xrightarrow{p} E[y_{2,j}-\eta(x_2,\theta)]=E_{\bar{y}_1}\{E[y_{2,j}-\eta(x_2,\theta)|\bar{y}_1]\}=0,\quad j=1,\dots,n_2,$$
by the weak law of large numbers. As $n\to\infty$, $n_2/n\to 1$, and so $j_n(\theta)/n\xrightarrow{p}[\dot{\eta}(x_2,\theta)]^2\sigma^{-2}$, and $j_n(\hat{\theta}_{n})/n\xrightarrow{p}[\dot{\eta}(x_2,\theta)]^2\sigma^{-2}=U^{-2}$.

\item
The first term of equation \eqref{iobsub} tends to $0$ when divided by $n$. In the second term, by the weak law of large numbers,
\begin{align*}
& \sum\limits^n_{i=n_1+1}\frac{[y_i-\eta(x_2,\theta)]^2/n_2}{\sigma^2}\xrightarrow{p}
\frac{1}{\sigma^2}E_{\bar{y}_1}
E\left\{[y_{n_1+1}-\eta(x_2,\theta)]^2\ \vert \ \bar{y}_1\right\}=1. 
\end{align*}
As $n\to\infty$, and $n_2/n\to 1$, $J^{\mathcal{D}}_n(\theta)/n\xrightarrow{p}[\dot{\eta}(x_2,\theta)]^2\sigma^{-2}$, and $J^{\mathcal{D}}_n(\hat{\theta}_{n})/n\xrightarrow{p}[\dot{\eta}(x_2,\theta)]^2\sigma^{-2}=U^{-2}$.

\item
The first term of equation \eqref{iobs} goes to $0$ when divided by $n$. In the second term, 
$\sqrt{n_2}[\bar{y}_2-\eta(x_2,\theta)]/\sigma$
 is distributed as $N(0,1)$ for every $n_2$, so $n_2[\bar{y}_2-\eta(x_2,\theta)]^2\sigma^{-2}\sim \chi^2(1)$ as $n_2\rightarrow\infty$. And $[\dot{\eta}(x_2,\theta)]^2\sigma^{-2}$ is independent of $n_2[\bar{y}_2-\eta(x_2,\theta)]^2\sigma^{-2}$. 
Therefore, $$n^{-1}J^{D}_n(\hat{\theta}_{n})\xrightarrow{d}[\dot{\eta}(x_2,\theta)]^2\sigma^{-2}n_2[\bar{y}_2-\eta(x_2,\theta)]^2\sigma^{-2}=U^{-2}W$$ where $W\sim \chi^2(1)$  as $n_2\rightarrow\infty$ and $n_2/n\rightarrow1$.

\item

Again, the first term of equations \eqref{iex} and \eqref{iexsub} go to $0$ when divided by $n$. As $n\to\infty$, and $n_2/n\to 1$, $I^{\mathcal{D}}_n(\theta)/n=I^{D}_n(\theta)/n\xrightarrow{p}[\dot{\eta}(x_2,\theta)]^2\sigma^{-2}$, and evaluated at $\theta=\hat{\theta_{n}}$, this limit is 
$U^{-2}$.
\end{enumerate}
\end{proof}

\subsection{The Generalized Cram{\'e}r-Slutzky theorem and Its Application}
Now we introduce the Generalized Cram{\'e}r-Slutzky Theorem in order to obtain main theoretical results in Theorem \ref{thm:main}, that is, to obtain  standard normal limits for $\hat{\theta}_{n}$ using  random norms. According to Lemma \ref{cong:rv}, the observed information $j_n(\theta)$, the stage-wise and subject-wise incremental expected information $I^{D}_n(\theta)$, $I^{\mathcal{D}}_n(\theta)$, and the subject-wise incremental observed information $J^{\mathcal{D}}_n(\theta)$ can be applied to normalize the MLE by the generalized Cram{\'e}r-Slutzky theorem, while the stage-wise incremental observed information $J^{D}_n(\theta)$ cannot.
\begin{lemma}{The Generalized Cram{\'e}r-Slutzky Theorem} \citep{aldous:1978}\label{con:function}
Suppose that $X_n\xrightarrow{d}X\quad \mathcal{G}-stably$. Let $g(x,y)$ be a continuous function of two variables, if $Y_{n}\xrightarrow{p}Y$, where $Y$ is a $\mathcal{G}$-measurable random variable. Then
$$g(X_n,Y_{n})\xrightarrow{d}g(X,Y)\quad \mathcal{G}-stably.$$ 
\end{lemma}

\begin{theorem}\label{thm:main}
Under model \eqref{eq:model}, $$j_n(\hat{\theta}_{n})^{1/2}(\hat{\theta}_{n}-\theta)\xrightarrow{d}N(0,1)\quad (mixing),$$ 
$$J^{\mathcal{D}}_n(\hat{\theta}_{n})^{1/2}(\hat{\theta}_{n}-\theta)\xrightarrow{d}N(0,1)\quad (mixing),$$
$$I^{D}_n(\hat{\theta}_{n})^{1/2}(\hat{\theta}_{n}-\theta)=I^{\mathcal{D}}_n(\hat{\theta}_{n})^{1/2}(\hat{\theta}_{n}-\theta)\xrightarrow{d}N(0,1)\quad (mixing),$$
as $n_2\rightarrow\infty$ with $n_1$  fixed.
\end{theorem}
\begin{proof}
Defining $g(x,y)=x\sqrt{y}$,  $g(x,y)$ is continuous function of two variables when $y>0$. Let $X_{n}=\sqrt{n}(\hat{\theta}_{n}-\theta)$ and $ Y_{n}=n^{-1}j_n(\hat{\theta}_{n})$.   Then $X_{n}\xrightarrow{d}UZ\quad \mathcal{F}_{n_1}-stably$ and $\ Y_{n}\xrightarrow{p}U^{-2}$. 
Because $U^{-2}=[\dot{\eta}(x_2,\theta)]^2\sigma^{-2}>0\ a.s$, $\sqrt{Y_{n}}\xrightarrow{p}U^{-1}$. Now by Lemma \ref{con:function}, $$g(X_{n},Y_{n})=\sqrt{n}(\hat{\theta}_{n}-\theta)n^{-1/2}j_n(\hat{\theta}_{n})^{1/2}=j_n(\hat{\theta}_{n})^{1/2}(\hat{\theta}_{n}-\theta)\xrightarrow{d}UZU^{-1}=Z\quad \mathcal{F}_{n_1}-stably.$$
Since $Z$ is independent of $\mathcal{F}_{n_1}$, 
$$j_n(\hat{\theta}_{n})^{1/2}(\hat{\theta}_{n}-\theta)\xrightarrow{d}N(0,1)\quad (mixing).$$
Similarly,  
$$J^{\mathcal{D}}_n(\hat{\theta}_{n})^{1/2}(\hat{\theta}_{n}-\theta)\xrightarrow{d}N(0,1)\quad (mixing),$$
$$I^{D}_n(\hat{\theta}_{n})^{1/2}(\hat{\theta}_{n}-\theta)=I^{\mathcal{D}}_n(\hat{\theta}_{n})^{1/2}(\hat{\theta}_{n}-\theta)\xrightarrow{d}N(0,1)\quad (mixing).$$
\end{proof}

\section{Adaptive Optimal Design Examples}
In this section, we apply Theorem \ref{thm:main} to normalize MLEs following an adaptive optimal design under logistic and  exponential (location and scale) regression models. Then we compare their tail probabilities and the difference between cumulative distribution functions using random norms and the Fisher information.  
For all models, the dose in the first stage is fixed at $x_1=2$, while the dose for stage~2 is selected from the range $x_2\in[a,b]=[0.25,4]$ based on stage~1 data. The divergence of the MLE of $\theta$ to infinity necessitates restricting  the search to  some finite interval  $(\underline{\theta},\overline{\theta})$; for simplicity throughout this section, we assume $\theta\in(0,1/a)$. All simulations assume the true parameter $\theta=1$ and known variance $\sigma^2=0.5^2$.

 The stage-two dose that maximizes the increase in information on the unknown parameter is 
\begin{align}\label{x2}
x_2(\theta)=\arg\max\limits_{x_2\in[a,b]}\ \frac{1}{n_2}Var(S_{n}-S_{n_{1}})=\arg\max\limits_{x_2\in[a,b]}\ [\dot{\eta}(x_2,\theta)]^2.\end{align}
 The two-stage adaptive optimal design  is $\xi_{A}=\{[w_{1},x_{1}],[w_{2},\hat{x}_{2}]\}$, where 
$\hat{x}_{2}$ is selected adaptively as given by \eqref{x2}, i.e., 
\begin{align*}
 \hat{x}_{2}=\begin{cases}a &\textrm{if }x_{2}(\hat{\theta}_{n_1})\le a;\\[-4pt] 
x_{2}(\hat{\theta}_{n_1}) &\textrm{if }x_{2}(\hat{\theta}_{n_1})\in [a,b];\\[-4pt]
b &\textrm{if }x_{2}(\hat{\theta}_{n_1})\ge b;
\end{cases}
\end{align*} 
and $w_{i}=n_{i}/n$.
 For each model, we evaluate the MLE norms' performance for several fixed values of $n_1$, including a   \emph{locally optimal stage~1 sample size}  \citep{lane:flou:info:2014}:   $$n_{1}^{*}(\theta)\equiv\arg\max\limits_{n_{1}\in\{1,...,n\}}i(\xi_{A},\theta),$$  
 where the notation $i(\theta)\equiv i(\xi_{A},\theta) $ makes Fisher information's dependence on the design explicit. To provide an ideal benchmark,  $n_{1}^{*}(\theta)$ is evaluated at the true value of $\theta$ for all models. A practical method to approximate the locally optimal stage~1 sample size is discussed by \citet{lane:flou:info:2014}.

\subsection{Logistic Regression Models}
 We explore the sample size needed to obtain the normal tail probabilities for the location parameter and scale parameter  logistic regression models, separately.
 
\subsubsection{The Logistic-Location Model}
Consider the  Logistic-Location Model with independent normal errors:
$$y=[1+e^{x-\theta}]^{-1}+\epsilon,\quad\epsilon\sim N(0,\sigma^2), \quad x\in[a,b],\quad -\infty < a<b<\infty$$ 
Maximizing the first-stage likelihood function,
$$\mathcal{L}_{n_1}(\theta|y_{11},\ldots,y_{1,n_1})
\propto\exp\left\{-\tfrac{n_1}{2\sigma^2}[\bar{y}_{1}-(1+e^{x_1-\theta})^{-1}]^2\right\},$$
yields the MLE:
\[\hat{\theta}_{n_1} = \left\{ \begin{array}{ll}
x_1+\textrm{logit}(\bar{y}_{1}) & \mbox{if $\bar{y}_{1}\in[0,(1+e^{x-\bar{\theta}})^{-1}]$}, \\0 & \mbox{if $\bar{y}_{1}\leq 0$}, \\ \bar{\theta} & \mbox{if $\bar{y}_{1}\geq (1+e^{x-\bar{\theta}})^{-1}$}.
\end{array}\right.\]
Adaptively selecting the second-stage dose to be
\[\hat{x}_{2} = \left\{ \begin{array}{ll}
\hat{\theta}_{n_1}  & \mbox{if $\bar{y}_{1}\in[(1+e^{x-a})^{-1},(1+e^{x-b})^{-1}]$}, \\b & \mbox{if $\bar{y}_{1}\geq (1+e^{x-b})^{-1}$}, \\ a & \mbox{if $\bar{y}_{1}\leq (1+e^{x-a})^{-1}$},
\end{array}\right.\]
the likelihood given data from both stages is
$$\mathcal{L}_n(\theta|y_{11},\ldots,y_{1,n_1},y_{21},\ldots,y_{2,n_2})
\propto\exp\left\{-\tfrac{n_1}{2\sigma^2}[\bar{y}_{1}-(1+e^{x_1-\theta})^{-1}]^2
-\tfrac{n_2}{2\sigma^2}[\bar{y}_{2}-(1+e^{x_{2}(\bar{y}_{1})-\theta})^{-1}]^2\right\},$$
and the MLE based on all data is 
\[\hat{\theta}_{n} = \left\{ \begin{array}{ll}
\theta_{n}^{'}  & \mbox{if $\theta_{n}^{'}\in(0,1/a)$}, \\0 & \mbox{if $\theta_{n}^{'}\leq 0$}, \\ 1/a & \mbox{if $\theta_{n}^{'}\geq 1/a$},
\end{array}\right.\]
where $\theta_{n}^{'}$ maximizes 
$\mathcal{L}_n(\theta|y_{11},\ldots,y_{1,n_1},y_{21},\ldots,y_{2,n_2}).$
The average Fisher information given data from both stages is
\begin{multline*}
\frac{1}{n}i(\xi_{A},\theta)=\frac{1}{\sigma^2}\bigg\{w_{1}(1+e^{x_1-\theta})^{-4}e^{2(x_1-\theta)}+w_{2}\pi_{a}(1+e^{a-\theta})^{-4}e^{2(a-\theta)}+w_{2}\pi_{b}(1+e^{b-\theta})^{-4}e^{2(b-\theta)}\\
+w_{2}E_{\bar{y}_{1}}\big[(1+e^{x_2-\theta})^{-4}e^{2(x_2-\theta)}I((1+e^{x_2-a})^{-1}<\bar{y}_{1}<(1+e^{x_2-b})^{-1})\big]\bigg\},
\end{multline*}
where $\pi_{a}=\Phi\{\sqrt{n_1}[(1+e^{x_2-a})^{-1}-(1+e^{x_2-\theta})^{-1}]/\sigma\}$ and 
$\pi_{b}=1-\Phi\{\sqrt{n_1}[(1+e^{x_2-b})^{-1}-(1+e^{x_2-\theta})^{-1}]/\sigma\}$ are the probabilities that $\hat{x}_{2}$ falls on the boundaries $a$ and $b$, respectively. 

According to functions \eqref{obs}, \eqref{iobsub} and \eqref{iex}, respectively,  the observed information is
\begin{multline*}
j_n(\theta)
=\frac{n_1}{\sigma^2}(1+e^{x_1-\theta})^{-4}e^{2(x_1-\theta)}
-\frac{n_1}{\sigma^2}\left[\bar{y}_1-(1+e^{x_1-\theta})^{-1}\right](1+e^{x_1-\theta})^{-3}(e^{x_1-\theta}-1)\\
+\frac{n_2}{\sigma^2}(1+e^{x_2-\theta})^{-4}e^{2(x_2-\theta)}
-\frac{n_2}{\sigma^2}\left[\bar{y}_2-(1+e^{x_2-\theta})^{-1}\right](1+e^{x_2-\theta})^{-3}(e^{x_2-\theta}-1);
\end{multline*}
the subject-wise incremental observed information is
\begin{align*}
J_n^{\mathcal{D}}(\theta)
&=\sum\limits^{n_1}_{i=1}\frac{[y_i-(1+e^{x_1-\theta})^{-1}]^2}{\sigma^4}(1+e^{x_1-\theta})^{-4}e^{2(x_1-\theta)}+\sum\limits^n_{i=n_1+1}\frac{[y_i-(1+e^{x_2-\theta})^{-1}]^2}{\sigma^4}(1+e^{x_2-\theta})^{-4}e^{2(x_2-\theta)};
\end{align*}
and the stage-wise and subject-wise incremental expected information are
\begin{align*}
I^{D}_n(\theta)=I^{\mathcal{D}}_n(\theta)
&=\frac{{n}_{1}}{\sigma^2}(1+e^{x_1-\theta})^{-4}e^{2(x_1-\theta)}+\frac{{n}_{2}}{\sigma^2}(1+e^{x_2-\theta})^{-4}e^{2(x_2-\theta)}.
\end{align*}

Table \ref{t1} provides  tail probabilities of MLE $\hat{\theta}_{n}$ normalized by the four different information measures when the first-stage sample size  is fixed at $n_1=30$ and the total sample size is $n=\{100,200,400\}$. In each scenario, \num[group-separator={,}]{10000} data sets are generated. When normalized by the expected (Fisher) information, MLEs' tail probabilities are far from the nominal ones in all situations. Random  information measures perform better in all scenarios, and as $n$ increases, the tail probabilities are closer to nominal ones. When $n=400$, the tail probabilities of MLE normalized by random information are almost the same as the nominal ones. 

Table \ref{t2} also provides tail probabilities of MLEs normalized by different information measures, but now  the first stage sample size is $n_{1}^{*}$ with total sample size $n=\{100,200,400\}$. The performance of the observed information does not change too much when compared with results in Table \ref{t1}, while other information measures (including the expected information) perform better. Still random information measures perform better than the Fisher information in all scenarios. Of course, as $n$ increases, the tail probabilities   normalized by all information measures are closer to the nominal one. When $n=400$ and $n_1=30$, the tail probabilities obtained with random information measures are rather close to nominal ones.

Figure \ref{f1} shows the integrated absolute difference between the cumulative distribution functions (CDF) of standard normal distribution and 
the  CDFs of MLE normalized by each of the four information measures when $n_1=n_{1}^{*}$. The integrated absolute difference between the CDFs of the t-distribution with 60 degrees of freedom and the standard normal distribution is also graphed to provide a sense of scale. One can see that normalizing MLEs with the random information measures brings them closer to the standard normal than  normalizing with the expected information. Moreover, the observed information performs best in terms of the integrated absolute difference when the total sample size $n$ is small. Of course, the distribution of MLEs with all normings becomes closer to the standard normal distribution  as $n$ increases.

\begin{table}[h!]
\caption{Tail probabilities of the MLE normalized with four different information measures under the logistic-location model}
\vspace*{0.1in}
\centering
{
\begin{tabular}{lccccc}
\hline\hline
 \multicolumn{1}{c}{$n_1/n$}  & \multicolumn{1}{c}{Information Measure} & \multicolumn{4}{c}{Left Tail / Right Tail} \\
 \hline
& Nominal &0.005/0.005 &	0.025/0.025 &	0.050/0.050 &	0.100/0.100 \\
\hline

& $i(\xi_{A},\theta)$ & 0.005/0.007 & 0.024/0.032 & 0.052/0.058 & 0.104/0.109 \\
{$30/100$}
& $j_n(\theta)$  & 0.004/0.006 & 0.023/0.030 & 0.049/0.055 & 0.098/0.106 \\
& $J_n^{\mathcal{D}}(\theta)$ & 0.004/0.007 & 0.022/0.029 & 0.047/0.052 & 0.095/0.103 \\
& $I_n^{\mathcal{D}}(\theta)/I_n^{D}(\theta)$ & 0.004/0.006 & 0.022/0.030 & 0.048/0.053 & 0.097/0.104 \\
\hline
& $i(\xi_{A},\theta)$ & 0.008/0.007 & 0.030/0.030 & 0.058/0.056 & 0.110/0.106 \\
{$30/200$}
& $j_n(\theta)$  & 0.006/0.006 & 0.025/0.026 & 0.051/0.050 & 0.101/0.099 \\
& $J_n^{\mathcal{D}}(\theta)$ & 0.006/0.006 & 0.025/0.025 & 0.050/0.050 & 0.099/0.098 \\
& $I_n^{\mathcal{D}}(\theta)/I_n^{D}(\theta)$ & 0.006/0.006 & 0.025/0.025 & 0.051/0.049 & 0.100/0.098 \\
\hline
& $i(\xi_{A},\theta)$ & 0.008/0.006 & 0.032/0.028 & 0.058/0.057 & 0.112/0.108 \\
{$30/400$}
& $j_n(\theta)$  & 0.006/0.004 & 0.026/0.024 & 0.050/0.050 & 0.102/0.100 \\
& $J_n^{\mathcal{D}}(\theta)$ & 0.006/0.004 & 0.027/0.024 & 0.050/0.049 & 0.100/0.099 \\
& $I_n^{\mathcal{D}}(\theta)/I_n^{D}(\theta)$ & 0.006/0.004 & 0.027/0.024 & 0.050/0.049 & 0.102/0.099 \\
\hline
\end{tabular}
}\vspace{12pt}
\label{t1}

\caption{Tail probabilities of the MLE normalized with four different information measures under the logistic-location model}
\vspace*{0.1in}
\centering
{
\begin{tabular}{lccccc}
\hline\hline
 \multicolumn{1}{c}{$n_1^{*}/n$}  & \multicolumn{1}{c}{Information Measure} & \multicolumn{4}{c}{Left Tail / Right Tail} \\
  \hline
& Nominal &0.005/0.005 &	0.025/0.025 &	0.050/0.050 &	0.100/0.100 \\
\hline
& $i(\xi_{A},\theta)$ & 0.013/0.010 & 0.035/0.035 & 0.060/0.060 & 0.111/0.111\\
{$17/100$}
& $j_n(\theta)$  & 0.005/0.008 & 0.025/0.030 & 0.050/0.055 & 0.098/0.105 \\
& $J_n^{\mathcal{D}}(\theta)$ & 0.004/0.008 & 0.031/0.029 & 0.054/0.051 & 0.103/0.099 \\
& $I_n^{\mathcal{D}}(\theta)/I_n^{D}(\theta)$ & 0.005/0.007 & 0.030/0.030 & 0.055/0.054 & 0.104/0.102 \\
\hline
& $i(\xi_{A},\theta)$ & 0.010/0.009	& 0.031/0.033	& 0.057/0.061	& 0.112/0.113 \\
{$23/200$}
& $j_n(\theta)$  & 0.005/0.006	& 0.024/0.029	& 0.047/0.054	& 0.100/0.105 \\
& $J_n^{\mathcal{D}}(\theta)$ & 0.009/0.006	& 0.028/0.028	& 0.050/0.053	& 0.101/0.102 \\
& $I_n^{\mathcal{D}}(\theta)/I_n^{D}(\theta)$ & 0.008/0.006	& 0.027/0.028	& 0.050/0.054	& 0.102/0.104 \\
\hline
& $i(\xi_{A},\theta)$ & 0.008/0.006 & 0.032/0.028 & 0.058/0.057 & 0.112/0.108 \\
{$30/400$}
& $j_n(\theta)$  & 0.006/0.004 & 0.026/0.024 & 0.050/0.050 & 0.102/0.100 \\
& $J_n^{\mathcal{D}}(\theta)$ & 0.006/0.004 & 0.027/0.024 & 0.050/0.049 & 0.100/0.099 \\
& $I_n^{\mathcal{D}}(\theta)/I_n^{D}(\theta)$ & 0.006/0.004 & 0.027/0.024 & 0.050/0.049 & 0.102/0.099 \\
\hline
\end{tabular}
}
\label{t2}
\end{table}

\begin{figure*}[h!]
\centerline{\includegraphics[height=.4\textheight]{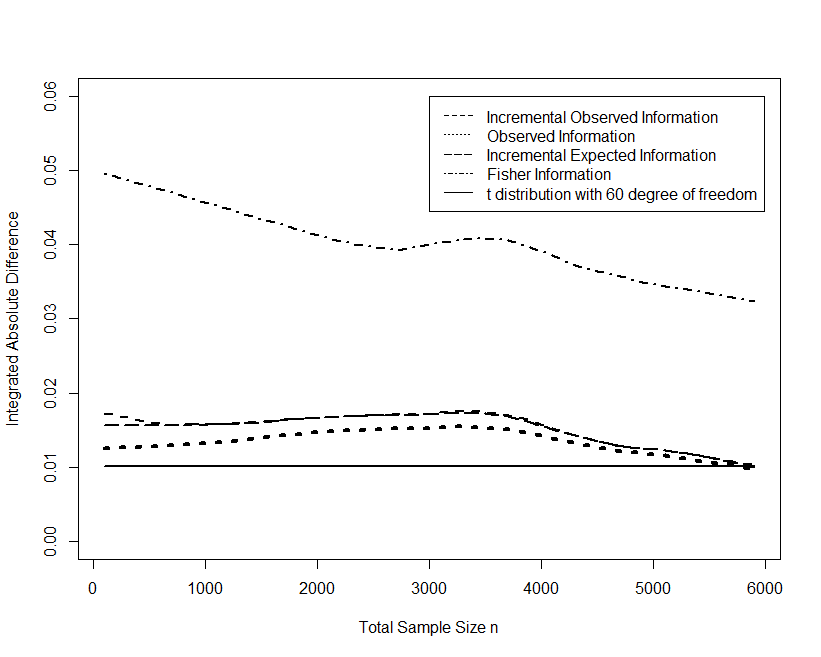}}
\caption{Integrated absolute differences  between the cumulative distributions of MLEs (1) normalized by four information measures and (2) the $t$-distribution with 60 degrees of freedom, and the cumulative distribution of standard normal distribution as $n$ increases from $100$ to \num[group-separator={,}]{6000} with $n_1=n_{1}^{*}$ under the logistic-location model.}
\label{f1}
\end{figure*}

\subsubsection{The Logistic-Scale  Model}
 To explore whether  the relative performance of different  information norms  differs between scale and location parameter estimates, now consider the Logistic-Scale Model:
\begin{align}\label{modleb}
 y=[1+e^{\theta x}]^{-1}+\epsilon,\quad\epsilon\sim N(0,\sigma^2),\quad x\in[a,b],\quad -\infty < a<b<\infty,
 \end{align}
 with first-stage  likelihood  
$\mathcal{L}_{n_1}(\theta|y_{11})
\propto\exp\left\{-\tfrac{n_1}{2\sigma^2}[\bar{y}_{1}-(1+e^{\theta x_1})^{-1}]^2\right\}.$
Because $\theta\in(0,\bar{\theta})$, maximizing the first-stage likelihood function yields the MLE
\[\hat{\theta}_{n_1} = \left\{ \begin{array}{ll}
\frac{-\textrm{logit}(\bar{y}_{1})}{x_1} & \mbox{if $\bar{y}_{1}\in\big((1+e^{\bar{\theta}x})^{-1},1/2\big)$}, \\0 & \mbox{if $\bar{y}_{1}\geq 1/2$}, \\ \bar{\theta} & \mbox{if $\bar{y}_{1}\leq (1+e^{\bar{\theta}x})^{-1}$}.
\end{array}\right.\]
The adaptively selected second-stage dose when $x_2(\hat\theta)\in[a,b]$ is
$$\hat{x}_{2}=\arg\max\big\{[\tfrac{d}{d\theta}\eta(x_{2},\theta)]^2|_{\theta=\hat{\theta}_{n_1}}\big\},$$
where 
$[\tfrac{d}{d\theta}\eta(x_{2},\theta)]^2=x_2^2e^{2\theta x_2}(1+e^{\theta x_2})^{-4}.$ The equation
$\frac{d}{dx_2}\left\{[\tfrac{d}{d\theta}\eta(x_{2},\theta)]^2|_{\theta=\hat{\theta}_{n_1}}\right\}=0$ has
no analytic solution 
for $x_2$, so we find it using  numerical methods. The MLE using all data can be found analytically: 
\[\hat{\theta}_{n} = \left\{ \begin{array}{ll}
\theta_{n}^{'}  & \mbox{if $\theta_{n}^{'}\in(0,1/a)$}, \\0 & \mbox{if $\theta_{n}^{'}\leq 0$}, \\ 1/a & \mbox{if $\theta_{n}^{'}\geq 1/a$},
\end{array}\right.\]
where $\theta_{n}^{'}$ maximizes
$$\mathcal{L}_n(\theta|y_{11},\ldots,y_{1,n_1},y_{21},\ldots,y_{2,n_2})
\propto\exp\left\{-\tfrac{n_1}{2\sigma^2}[\bar{y}_{1}-(1+e^{\theta x_1})^{-1}]^2
-\tfrac{n_2}{2\sigma^2}[\bar{y}_{2}-(1+e^{\theta x_{2}(\bar{y}_{1})})^{-1}]^2\right\}.$$
The Fisher information is derived numerically because we do not have an explicit function for $x_2(\bar{y}_{1})$. However, we can still derive the random information measures analytically. The observed information is
\begin{multline*}
j_n(\theta)
=\frac{n_1}{\sigma^2}x_1^2e^{2\theta x_1}(1+e^{\theta x_1})^{-4}
-\frac{n_1}{\sigma^2}\left[\bar{y}_1-(1+e^{\theta x_1})^{-1}\right](1+e^{\theta x_1})^{-3}(e^{\theta x_1}-1)e^{\theta x_1}x_1^2\\
+\frac{n_2}{\sigma^2}x_2^2e^{2\theta x_2}(1+e^{\theta x_2})^{-4}
-\frac{n_2}{\sigma^2}\left[\bar{y}_2-(1+e^{\theta x_2})^{-1}\right](1+e^{\theta x_2})^{-3}(e^{\theta x_2}-1)e^{\theta x_2}x_2^2;
\end{multline*}
the subject-wise incremental observed information is
\begin{align*}
J_n^{\mathcal{D}}(\theta)
&=\sum\limits^{n_1}_{i=1}\frac{[y_i-(1+e^{\theta x_1})^{-1}]^2}{\sigma^4}x_1^2e^{2\theta x_1}(1+e^{\theta x_1})^{-4}+\sum\limits^n_{i=n_1+1}\frac{[y_i-(1+e^{\theta x_2})^{-1}]^2}{\sigma^4}x_2^2e^{2\theta x_2}(1+e^{\theta x_2})^{-4};
\end{align*}
and the stage-wise and subject-wise incremental expected information are
\begin{align*}
I^{D}_n(\theta)=I^{\mathcal{D}}_n(\theta)
&=\frac{{n}_{1}}{\sigma^2}x_1^2e^{2\theta x_1}(1+e^{\theta x_1})^{-4}+\frac{{n}_{2}}{\sigma^2}x_2^2e^{2\theta x_2}(1+e^{\theta x_2})^{-4}.
\end{align*}

Table~\ref{t3} presents the tail probabilities of the MLE  $\hat{\theta}_{n}$ normalized by these four different information measures under Model \ref{modleb}. First, comparing with the same scenario in Table \ref{t3}, we note that the  tail probabilities are much farther from the nominal ones
when the first-stage sample~size is fixed at $n_1=20$ and $n=400$.  This indicates that the scale-parameter estimators converge more slowly than the location-parameter estimators. Because estimators perform rather poorly when $n_1$ is small regarding of their normalizing measure, Table~\ref{t3} also shows tail probabilities  for larger sample sizes, namely, $n_1=100$ and $n=\{2000, 3000\}$. In both cases, random information measures perform better than the expected information. As $n$ increases, the tail probabilities of the MLE normalized by random information measures are closer to the nominal ones. When $n=3000$, the tail probabilities of the MLE normalized by random information measures are almost identical to nominal ones.

The tail probabilities of MLEs normalized by different information measures when  the first stage sample size is $n_{1}^{*}$, and $n=\{400,2000,3000\}$ are presented in Table \ref{t4}. The performance of random information measures does not change  much when compared with results in Table \ref{t3}, while the expected information perform better. However, random information measures still perform as good as the expected information in both $n=2000$ and $n=3000$ cases. Still, as $n$ increases, the tail probabilities  normalized by all information measures are closer to the nominal one. When $n=3000$ and $n_1=434$, the tail probabilities obtained with all information measures are rather close to nominal ones.

Figure \ref{f2} shows the integrated absolute difference between the cumulative distribution functions (CDFs) of MLE normalized by the four information measures and the CDF of standard normal distribution when $n_1=n_{1}^{*}$. Regarding the integrated absolute difference between the CDF of t-distribution with 60 degree of freedom and standard normal distribution as a benchmark, one can see that normalizing MLEs with each information measure brings them closer to being standard normal as $n$ increases. In addition, the distribution of MLEs normalized by each random information measure is as close to the standard normal distribution as the one normalized by the expected information. The reason is that $n_1=n_{1}^{*}$ under this model is large enough that the MLE converges to a normal distribution instead of a random scale mixture of normal distribution.

\begin{table}[h!]
\caption{Tail probabilities of the MLE  normalized with four different information measures under the
logistic-scale Model}
\vspace*{0.1in}
\centering
{
\begin{tabular}{lccccc}
\hline\hline
 \multicolumn{1}{c}{$n_1/n$}  & \multicolumn{1}{c}{Information Measure} & \multicolumn{4}{c}{Left Tail / Right Tail} \\
 \hline
& Nominal &0.005/0.005 &	0.025/0.025 &	0.050/0.050 &	0.100/0.100 \\
\hline

& $i(\xi_{A},\theta)$ & 0.010/0.016 & 0.033/0.156 & 0.055/0.176 & 0.098/0.216 \\
{$30/400$}
& $j_n(\theta)$  & 0.013/0.002 & 0.041/0.015 & 0.065/0.034 & 0.107/0.079 \\
& $J_n^{\mathcal{D}}(\theta)$ & 0.012/0.001 & 0.040/0.013 & 0.063/0.030 & 0.106/0.073 \\
& $I_n^{\mathcal{D}}(\theta)/I_n^{D}(\theta)$ & 0.012/0.002 & 0.040/0.012 & 0.064/0.031 & 0.106/0.073 \\
\hline
& $i(\xi_{A},\theta)$ & 0.008/0.024 & 0.030/0.051 & 0.053/0.074 & 0.102/0.121 \\
{$100/2000$}
& $j_n(\theta)$  & 0.009/0.003 & 0.031/0.020 & 0.054/0.046 & 0.106/0.093 \\ 
& $J_n^{\mathcal{D}}(\theta)$ & 0.009/0.003 & 0.031/0.020 & 0.054/0.045 & 0.105/0.092 \\ 
& $I_n^{\mathcal{D}}(\theta)/I_n^{D}(\theta)$ & 0.009/0.003 & 0.030/0.020 & 0.054/0.045 & 0.106/0.092 \\
\hline
& $i(\xi_{A},\theta)$ & 0.007/0.032 & 0.027/0.049 & 0.052/0.076 & 0.100/0.125 \\
{$100/3000$}
& $j_n(\theta)$  & 0.007/0.004 & 0.028/0.023 & 0.053/0.049 & 0.104/0.101 \\ 
& $J_n^{\mathcal{D}}(\theta)$ & 0.007/0.004 & 0.027/0.022 & 0.052/0.049 & 0.104/0.100 \\ 
& $I_n^{\mathcal{D}}(\theta)/I_n^{D}(\theta)$ & 0.007/0.004 & 0.028/0.023 & 0.053/0.048 & 0.103/0.100 \\
\hline

\end{tabular}
}\vspace{12pt}
\label{t3}

\caption{Tail probabilities of the MLE  normalized with four different information measures under the 
logistic-scale Model}
\vspace*{0.1in}
\centering
{
\begin{tabular}{lccccc}
\hline\hline
 \multicolumn{1}{c}{$n_1^{*}/n$}  & \multicolumn{1}{c}{Information Measure} & \multicolumn{4}{c}{Left/Right} \\
  \hline
& Nominal &0.005/0.005 &	0.025/0.025 &	0.050/0.050 &	0.100/0.100 \\
\hline
& $i(\xi_{A},\theta)$ & 0.012/$<$0.001 & 0.037/0.006 & 0.063/0.025 & 0.112/0.072 \\
{$174/400$}
& $j_n(\theta)$  & 0.013/0.001 & 0.040/0.011 & 0.068/0.030 & 0.116/0.080 \\
& $J_n^{\mathcal{D}}(\theta)$ & 0.013/0.001 & 0.038/0.010 & 0.067/0.028 & 0.116/0.078 \\
& $I_n^{\mathcal{D}}(\theta)/I_n^{D}(\theta)$ & 0.013/$<$0.001 & 0.040/0.010 & 0.067/0.029 & 0.115/0.078 \\
\hline
& $i(\xi_{A},\theta)$ & 0.006/0.002	& 0.028/0.020	& 0.058/0.045	& 0.105/0.092 \\
{$361/2000$}
& $j_n(\theta)$  & 0.007/0.003	& 0.030/0.020	& 0.060/0.046	& 0.107/0.094 \\
& $J_n^{\mathcal{D}}(\theta)$ & 0.007/0.003	& 0.030/0.020	& 0.059/0.046	& 0.106/0.094 \\
& $I_n^{\mathcal{D}}(\theta)/I_n^{D}(\theta)$ & 0.007/0.002	& 0.029/0.020	& 0.059/0.046	& 0.106/0.093 \\

\hline
& $i(\xi_{A},\theta)$ & 0.007/0.003	& 0.029/0.020	& 0.053/0.044	& 0.102/0.097 \\
{$434/3000$}
& $j_n(\theta)$  & 0.008/0.003	& 0.029/0.022	& 0.054/0.045	& 0.104/0.097 \\
& $J_n^{\mathcal{D}}(\theta)$  & 0.007/0.004	& 0.029/0.022	& 0.053/0.044	& 0.103/0.097 \\
& $I_n^{\mathcal{D}}(\theta)/I_n^{D}(\theta)$ & 0.007/0.003	& 0.029/0.021	& 0.053/0.045	& 0.103/0.097 \\

\hline

\end{tabular}
}
\label{t4}
\end{table}

\begin{figure*}[h!]
\centerline{\includegraphics[height=.4\textheight]{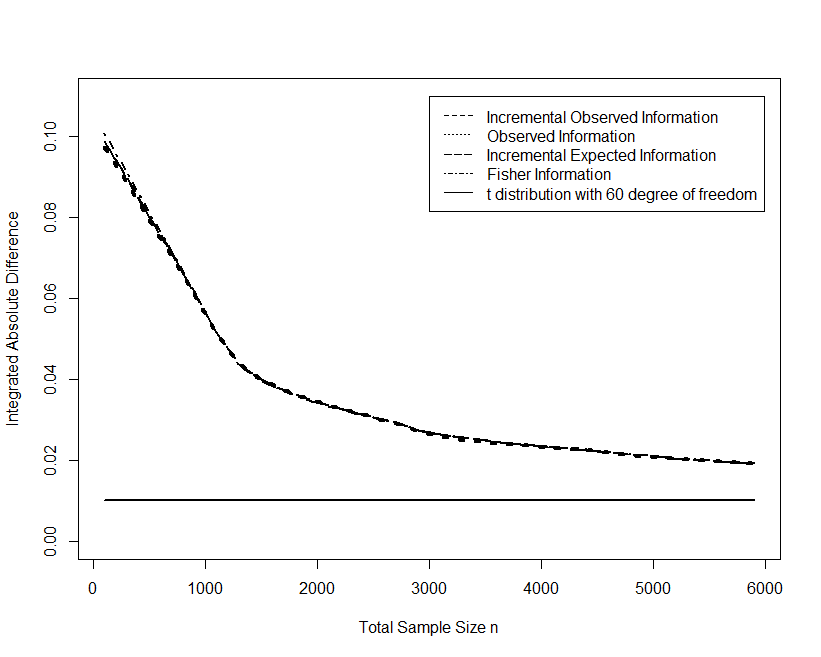}}
\caption{Integrated absolute differences  between the cumulative distributions of MLEs (1) normalized by four information measures and (2) the $t$-distribution with 60 degrees of freedom, and the cumulative distribution of standard normal distribution as $n$ increases from $100$ to $6000$ with $n_1=n_{1}^{*}$ under the logistic-scale model.}
\label{f2}
\end{figure*}

\subsection{The Exponential Regression Model}
Now we present and compare the performance of estimators from two-stage adaptive optimal designs under exponential location and scale models.
\subsubsection{The Exponential-Location  Model}
 The Exponential-Location Model,
$$y=e^{-x+\theta}+\epsilon,\quad\epsilon\sim N(0,\sigma^2),\quad x\in[a,b],\quad -\infty < a<b<\infty,$$
 has first-stage  MLE
\[\hat{\theta}_{n_1} = \left\{ \begin{array}{ll}
x_1+\log\bar{y}_{1} & \mbox{if $\bar{y}_{1}\in(e^{-x_1},e^{-x_1+\bar{\theta}})$}, \\0 & \mbox{if $\bar{y}_{1}\leq e^{-x_1}$}, \\ \bar{\theta} & \mbox{if $\bar{y}_{1}\geq e^{-x_1+\bar{\theta}}$}.
\end{array}\right.\]
The adaptively selected second-stage dose is
$$\hat{x}_{2}=\arg\max[\tfrac{d}{d\theta}\eta(x_{2},\theta)]^2|_{\theta=\hat{\theta}_{n_1}},\quad x_2\in[a,b];$$
Note that $[\tfrac{d}{d\theta}\eta(x_{2},\theta)]^2=(e^{-x_2+\theta})^2=e^{2(\theta-x_2)}$
 is maximized at $x_2=a$,  a constant that does not depend on $\bar{y}_{1}$. So under this model,  if $n_1/n$ is small (and provided common regularity conditions hold with $\dot{\eta}\neq 0,\ |\dot{\eta}|<\infty$), then $U=\sigma[\dot{\eta}(x_{2},\theta)]^{-1}$ in equation \eqref{eq:lane} is a constant, instead of a random function of $\bar{y}_{1}$, even if $n_1$ is held fixed while $n\rightarrow\infty$. In this case, this procedure is simply two separate designs, instead of two-stage adaptive design. Therefore, $\sqrt[]{n}(\hat{\theta}_{n}-\theta)\stackrel{d}{\rightarrow}N(0,U^2)$  as   $n_2\rightarrow\infty$  with $n_1$  fixed.

\subsubsection{The Exponential-Scale  Model}\hfill\\
Consider 
$$y=e^{-\theta x}+\epsilon,\quad\epsilon\sim N(0,\sigma^2), \quad x\in[a,b],\quad -\infty < a<b<\infty,$$
 $\theta=1$. In this model, \citet{lane:flou:info:2014} have the following results: The MLE based on data from first stage is
\[\hat{\theta}_{n_1} = \left\{ \begin{array}{ll}
\frac{-\log\bar{y}_{1}}{x_1} & \mbox{if $\bar{y}_{1}\in(e^{-x_{1}\bar{\theta}},1)$}, \\0 & \mbox{if $\bar{y}_{1}\geq 1$}, \\ \bar{\theta} & \mbox{if $\bar{y}_{1}\leq e^{-x_{1}\bar{\theta}}$}.
\end{array}\right.\]
The   second stage uses the dose 
\[\hat{x}_{2} = \left\{ \begin{array}{ll}
\hat{\theta}_{n_1}^{-1}  & \mbox{if $\bar{y}_{1}\in(e^{-x_{1}/a},e^{-x_{1}/b})$}, \\b & \mbox{if $\bar{y}_{1}\geq e^{-x_{1}/b}$}, \\ a & \mbox{if $\bar{y}_{1}\leq e^{-x_{1}/a}$}.
\end{array}\right.\]
\citet{lane:flou:two-:2012} find the average expected information for the Exponential-Scale Model to be 
\begin{align*}
\frac{1}{n}i(\xi_{A},\theta)=\frac{1}{\sigma^2}\Bigg\{ w_{1}x_{1}^{2}e^{-2\theta x_{1}}
&+w_{2}\pi_{a}a^{2}e^{-2\theta a}+w_{2}\pi_{b}b^{2}e^{-2\theta b} \\
  &+w_{2}E_{\bar{y}_{1}}\left[\left(\frac{-\log\bar{y}_{1}}{x_1}\right)^{2}e^{2\theta x_{1}/\log\bar{y}_{1}}I(e^{-x_{1}/a}<\bar{y}_{1}<e^{-x_{1}/b})\right]\Bigg\},\end{align*}
where $\pi_{a}=\Phi(\sqrt{n_1}(e^{-x_{1}/a}-e^{-x_{1}\theta})/\sigma)$ and 
$\pi_{b}=1-\Phi(\sqrt{n_1}(e^{-x_{1}/b}-e^{-x_{1}\theta})/\sigma)$. The MLE using all data is 
\[\hat{\theta}_{n} = \left\{ \begin{array}{ll}
\theta_{n}^{'}  & \mbox{if $\theta_{n}^{'}\in(0,1/a)$}, \\0 & \mbox{if $\theta_{n}^{'}\leq 0$}, \\ 1/a & \mbox{if $\theta_{n}^{'}\geq 1/a$},
\end{array}\right.\]
where $\theta_{n}^{'}$ maximizes
$\mathcal{L}_n(\theta|y_{11},\ldots,y_{1,n_1},y_{21},\ldots,y_{2,n_2})
\propto\exp\left\{-\tfrac{n_1}{2\sigma^2}[\bar{y}_{1}-e^{-\theta x_{1}}]^2
-\tfrac{n_2}{2\sigma^2}[\bar{y}_{2}-e^{-\theta x_{2}(\bar{y}_{1})}]^2\right\}.$

According to functions \eqref{obs}, \eqref{iobsub} and \eqref{iex},  the observed information is
\begin{align*}
j_n(\theta)
&=\frac{n_1}{\sigma^2}x_{1}^{2}e^{-2\theta x_{1}}
-\frac{n_1}{\sigma^2}\left(\bar{y}_1-e^{-\theta x_{1}}\right)x_{1}^{2}e^{-\theta x_{1}}
+\frac{n_2}{\sigma^2}x_{2}^{2}e^{-2\theta x_{2}}
-\frac{n_2}{\sigma^2}\left(\bar{y}_2-e^{-\theta x_{2}}\right)x_{2}^{2}e^{-\theta x_{2}};
\end{align*}
the subject-wise incremental observed information 
\begin{align*}
J_n^{\mathcal{D}}(\theta)
&=\sum\limits^{n_1}_{i=1}\frac{[y_i-e^{-\theta x_{1}}]^2}{\sigma^4}x_{1}^{2}e^{-2\theta x_{1}}+\sum\limits^n_{i=n_1+1}\frac{[y_i-e^{-\theta x_{2}}]^2}{\sigma^4}x_{2}^{2}e^{-2\theta x_{2}};
\end{align*}
and the stage-wise and subject-wise incremental expected information are
\begin{align*}
I^{D}_n(\theta)=I^{\mathcal{D}}_n(\theta)
&=\frac{{n}_{1}}{\sigma^2}x_{1}^{2}e^{-2\theta x_{1}}+\frac{{n}_{2}}{\sigma^2}x_{2}^{2}e^{-2\theta x_{2}}.
\end{align*}
Again, we use the random information measures  to normalize the MLE $\hat{\theta}_{n}$.

\begin{table}[h!]
\caption{Tail probabilities of MLE of scale parameter normalized by different information measures under the exponential-scale Model}
\vspace*{0.1in}
\centering
{
\begin{tabular}{lccccc}
\hline\hline
 \multicolumn{1}{c}{$n_1/n$}  & \multicolumn{1}{c}{Information Measure} & \multicolumn{4}{c}{Left Tail / Right Tail} \\
 \hline
& Nominal &0.005/0.005 &	0.025/0.025 &	0.050/0.050 &	0.100/0.100 \\
\hline
& $i(\xi_{A},\theta)$ & 0.013/0.037	& 0.039/0.056 &	0.068/0.080 &	0.117/0.128 \\
{$30/500$}
& $j_n(\theta)$  & 0.009/0.003	& 0.033/0.018 &	0.061/0.040 &	0.108/0.090 \\
& $J_n^{\mathcal{D}}(\theta)$ & 0.009/0.002 &	0.033/0.018 &	0.061/0.039 &	0.108/0.088 \\
& $I_n^{\mathcal{D}}(\theta)/I_n^{D}(\theta)$ & 0.009/0.002 &	0.032/0.017 &	0.061/0.039 &	0.108/0.088 \\

\hline
& $i(\xi_{A},\theta)$ & 0.012/0.038	& 0.036/0.060	& 0.062/0.085 &	0.114/0.132 \\
{$30/1000$}
& $j_n(\theta)$  & 0.007/0.003	& 0.029/0.021	& 0.055/0.047	& 0.106/0.094 \\
& $J_n^{\mathcal{D}}(\theta)$ & 0.007/0.003	& 0.028/0.020	& 0.055/0.046	& 0.104/0.093 \\
& $I_n^{\mathcal{D}}(\theta)/I_n^{D}(\theta)$ & 0.007/0.003	& 0.029/0.020	& 0.055/0.046	& 0.106/0.093 \\

\hline
& $i(\xi_{A},\theta)$ & 0.011/0.037	&	0.035/0.059	&	0.062/0.086	&	0.112/0.132 \\
{$30/2000$}
& $j_n(\theta)$  & 0.006/0.004	&	0.028/0.023	&	0.052/0.048	&	0.100/0.096 \\
& $J_n^{\mathcal{D}}(\theta)$ & 0.006/0.004	&	0.027/0.023	&	0.050/0.047	&	0.099/0.096 \\
& $I_n^{\mathcal{D}}(\theta)/I_n^{D}(\theta)$ & 0.006/0.004	&	0.027/0.023	&	0.052/0.048	&	0.099/0.096 \\

\hline
\end{tabular}
}\vspace{12pt}
\label{t5}
\caption{Tail probabilities of MLE of scale parameter normalized by different information measures under the exponential-scale Model}

\vspace*{0.1in}
\centering
{
\begin{tabular}{lccccc}
\hline\hline
 \multicolumn{1}{c}{$n_1^{*}/n$}  & \multicolumn{1}{c}{Information Measure} & \multicolumn{4}{c}{Left/Right} \\
  \hline
& Nominal &0.005/0.005 &	0.025/0.025 &	0.050/0.050 &	0.100/0.100 \\
\hline
& $i(\xi_{A},\theta)$ & 0.012/0.009	& 0.038/0.027	& 0.065/0.052	& 0.117/0.105 \\
{$60/500$}
& $j_n(\theta)$  & 0.010/0.003	& 0.033/0.018	& 0.060/0.040	& 0.111/0.096 \\
& $J_n^{\mathcal{D}}(\theta)$ & 0.010/0.003	& 0.033/0.017	& 0.058/0.039	& 0.110/0.095 \\
& $I_n^{\mathcal{D}}(\theta)/I_n^{D}(\theta)$ & 0.010/0.003	& 0.033/0.017	& 0.060/0.039	& 0.111/0.094 \\

\hline
& $i(\xi_{A},\theta)$ & 0.008/0.007	& 0.035/0.028	& 0.060/0.050	& 0.113/0.100 \\
{$86/1000$}
& $j_n(\theta)$  & 0.007/0.004	& 0.031/0.024	& 0.056/0.046	& 0.107/0.095 \\
& $J_n^{\mathcal{D}}(\theta)$ & 0.007/0.004	& 0.031/0.024	& 0.055/0.046	& 0.107/0.095 \\
& $I_n^{\mathcal{D}}(\theta)/I_n^{D}(\theta)$ & 0.007/0.004	& 0.031/0.024	& 0.056/0.046	& 0.107/0.095 \\

\hline
& $i(\xi_{A},\theta)$ & 0.007/0.004 & 0.030/0.024 & 0.058/0.049 & 0.107/0.101 \\
{$122/2000$}
& $j_n(\theta)$  & 0.005/0.004 & 0.027/0.022 & 0.054/0.047 & 0.103/0.098 \\
& $J_n^{\mathcal{D}}(\theta)$ & 0.005/0.004 & 0.027/0.022 & 0.053/0.046 & 0.103/0.097 \\
& $I_n^{\mathcal{D}}(\theta)/I_n^{D}(\theta)$ & 0.005/0.004 & 0.027/0.022 & 0.054/0.046 & 0.107/0.101 \\

\hline
\end{tabular}
}
\label{t6}
\end{table}

Because all information measures perform poorly when $n$ is smaller than $500$, we present tail probabilities of MLE normalized by different information measures when the first-stage sample~size is fixed at $n_1=30$ and $n=\{500, 1000, 2000\}$ in Table \ref{t5}. In each scenario, $10,000$ data sets are generated. The table shows that  MLEs standardized by random information measures have tail probabilities  closer to  the normal ones than is obtained using the expected information for all scenarios. All the random norms have similar convergence rates. When $n=2000$, the tail probabilities of MLE normalized by random information measures are almost the same as nominal ones. Furthermore, recall that the standard normal limit is not obtained by normalizing the MLE using the expected information.

Table \ref{t6} shows  tail probabilities of MLEs normalized by different information measures when the first stage sample size is $n_1=n_{1}^{*}$ and $n=\{500, 1000, 2000\}$. Note that the performance using random information measures does not change  much from what is shown in table~\ref{t5}, while the expected information performs better. When $n=2000$, all random information measures perform well on normalizing MLE.

Figure \ref{f3} shows MLEs normalized by random information measures are closer to standard normal than those normalized by the expected information in terms of integrated absolute difference as expected. Additionally, the observed information still performs best when the total sample size $n$ is small. And as $n$ increases, MLE with all information measures is closer to standard normal distribution.

\begin{figure*}[h!]
\centerline{\includegraphics[height=.4\textheight]{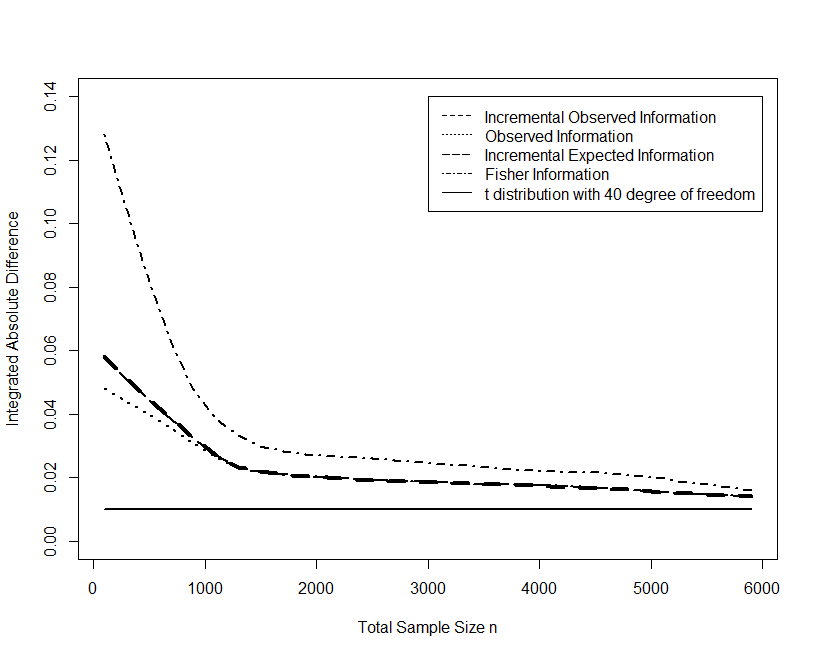}}
\caption{Integrated absolute differences  between the cumulative distributions of MLEs (1) normalized by four information measures and (2) the $t$-distribution with 60 degree of freedom, and the cumulative distribution of standard normal distribution as $n$ increases from $100$ to $6000$ with $n_1=n_{1}^{*}$ under the exponential-scale model.}
\label{f3}
\end{figure*}

\section{Discussion}
A major motivation for this work was the finding that using the Fisher information  to normalize the MLE (following   a two-stage optimal adaptive design under nonlinear regression models with independent normal errors) produces a normal variance mixture limiting distribution instead of a normal limit when the first stage sample size is held fixed [\citep{lane:flou:two-:2012}]. Although they found the asymptotic distribution of the MLE in the case of an exponential mean function, their derivation is not generalizable. This paper follows up on the comment by \citet{Barn:Sore:arev:1994} that replacing the Fisher information with random norms can yield normal limits. We show how this is done using the Generalized Cram{\'e}r-Slutzky theorem in the situation studied by \citet{lane:flou:two-:2012}. 

We also establish this result in the same situation. That is, we derived  the observed,  incremental expected and observed information measures in the case of a two-stage adaptive design under a general nonlinear regression model with conditionally independent normal errors and proved that using them to 
norm  the MLE 
yields standard normal distributions when the first sample size is fixed and the second stage sample size is large. 

We illustrate these findings assuming logistic and exponential mean functions, and compare  the estimation performance  using the three random information measures with the Fisher information norming. We found better performance using the observed information  than using other norms, including the Fisher information  under both models.

Additionally, we show that larger sample sizes are required to obtain normal tail probabilities for the scale parameter than for the location parameter under the logistic model. Moreover, the location parameter under the logistic model converges faster than the scale parameter under the exponential model.

It is important to recall from \citet{lane:flou:two-:2012} that the independence of $U$ and $Z$ in \eqref{eq:lane} results from the independence of the sample mean and standard deviation in the normal  error distribution of model \eqref{eq:model}. The consequences of changing the error distribution need to be investigated.

 A variety of adaptive methods are used in clinical trials, including, for example, enrichment designs,   early stopping for toxicity and/or efficacy, and  sample size re-estimation. We plan to explore  whether the methodology presented in this paper is applicable in some situations following these adaptive methods as well.
 
\section{Acknowledgement}
We thank Dave Mason and Erich H{\"a}usler for useful conversations on stable convergence. And we thank Dr. H{\"a}usler especially for sharing his expertise concerning the proof of Lemma 3.1.

\bibliography{references}

\end{document}